\documentclass{article}

\author{Wouter Kuijper, Victor Ermolaev, and Olivier Devillers}
\title{Celestial Walk: A Terminating Oblivious Walk for Convex Subdivisions}

\usepackage{amsmath}
\usepackage{amssymb}
\usepackage{amsthm}
\usepackage{tikz}
\usepackage{xspace}
\usetikzlibrary{calc,intersections}

\usepackage{listings}
\usepackage{algorithm}
\usepackage{algpseudocode}

\usepackage{hyperref}

\usepackage{caption}
\usepackage{subcaption}

\definecolor{BisectorPurple}{rgb}{0.68,0.42,0.83}

\definecolor{Grey}{rgb}{0.6,0.6,0.6}
\definecolor{OliveGreen}{rgb}{0.2,0.5,0.2}
\definecolor{VictorOrange}{rgb}{1,0.5,0.3}
\definecolor{KuijperCrimson}{rgb}{.5,0,0}

\newtheorem{theorem}{Theorem}

\newcommand{\ms}[1]{{\mathsf{#1}}}
\newcommand{\mc}[1]{{\mathcal{#1}}}
\newcommand{\mb}[1]{{\mathbb{#1}}}

\newcommand{\PSLG}{PSLG\xspace}
\newcommand{\PSLGs}{PSLGs\xspace}

\newcommand{\Next}[1]{{\ms{next}(#1)}}

\newcommand{\Twin}[1]{{\ms{twin}(#1)}}
\newcommand{\Origin}[1]{{\ms{origin}(#1)}}
\newcommand{\Target}[1]{{\ms{target}(#1)}}
\newcommand{\Face}[1]{{\ms{face}(#1)}}

\newcommand{\Dist}[1]{{\ms{cd}(#1)}}

\newcommand{\wideangle}[1]{{\ms{wideangle}(#1)}}
\newcommand{\narrowangle}[1]{{\ms{angle}(#1)}}

\newcommand{\Obtuse}[1]{\ensuremath{\ms{obtuse}(#1)}}
\newcommand{\ApproxBisector}[1]{\ensuremath{\ms{approx\_bisector}(#1)}}

\title{
Celestial Walk: A Terminating Oblivious Walk for Convex Subdivisions
}

\author{
Wouter Kuijper\\Nedap N.V.
\and
Victor Ermolaev\\Nedap N.V.
\and
Olivier Devillers\\Loria, Inria, CNRS, Universit\'e de Lorraine, France.}

\begin{document}

\maketitle

\begin{abstract}
We present a new oblivious walking strategy for convex
  subdivisions. Our walk is faster than the straight walk and more
  generally applicable than the visiblity walk. To prove termination
  of our walk we use a novel monotonically decreasing distance
  measure.
\end{abstract}

\begin{figure}\center
\includegraphics[page=1,width=0.8\textwidth]{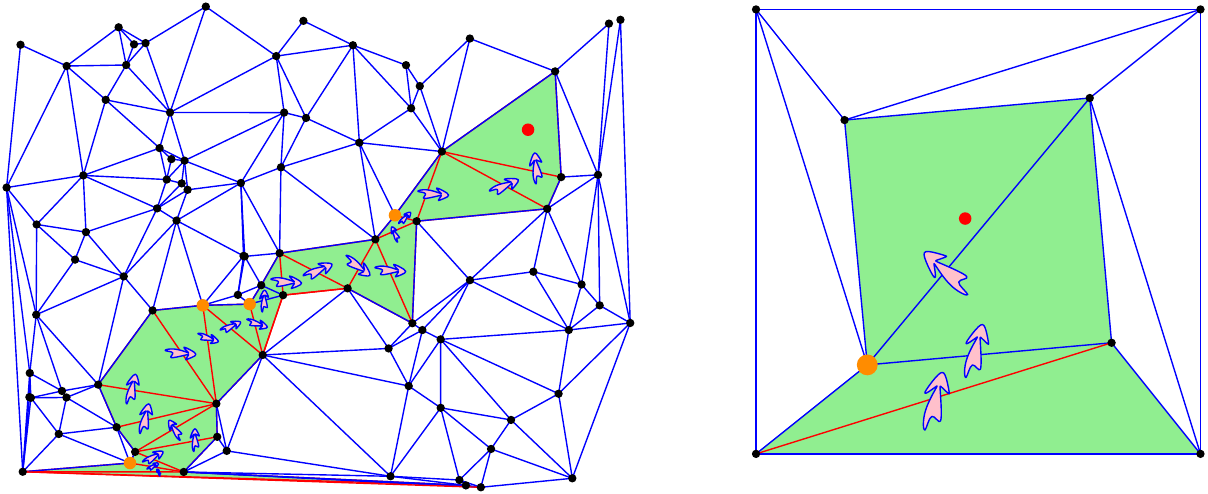}
\label{fig:celestial}
\caption{Celestial walk (obtuse angles marked by orange dots).}
\end{figure}

\section{Introduction}

Point location in a convex subdivision is a classical problem of
computational geometry for which several data structures
have been designed with good complexities in the worst
case~\cite{kirkpatrick1983optimal,chazelle1986fractional,preparata1990planar}.
These intricate solutions are often unused in favor of simpler
algorithms based on traversal of the planar subdivisions using
neighborhood relations between faces, also known as
{\em walking algorithms}~\cite{bose2002online,bose2004online,devillers:inria-00102194}.
These walking algorithms can also be used as a building block
in randomized data structures for point
location~\cite{mucke1999fast,devillers:inria-00166711}.

Amongst convex subdivisions, Delaunay triangulations received a lot of
attention because of their practical importance. For Delaunay
triangulations, essentially two walking strategies are used: the
\emph{straight walk} and the \emph{visibility
  walk}~\cite{devillers:inria-00102194}.  The straight walk visits all
faces crossed by a line segment between a known face and the query
point, while the visibility walk goes from a face to another if the
query point is on the side of the new face with respect to the
supporting line of the edge common to the two faces (cf. Figure~\ref{fig:classics}).

The straight walk trivially terminates in the face containing the
query point and generalizes to any planar subdivision but with the
inconvenience of not being oblivious: the starting face of the walk
must be remembered during the whole walk.

The visibility walk is oblivious, but proving its termination requires
the use of particular properties of the Delaunay triangulation, and
actually the visibility walk may loop in other
subdivisions~\cite{devillers:inria-00102194} (cf. Figure~\ref{fig:loop}).

\begin{figure}[t]
\includegraphics[page=2,width=1\textwidth]{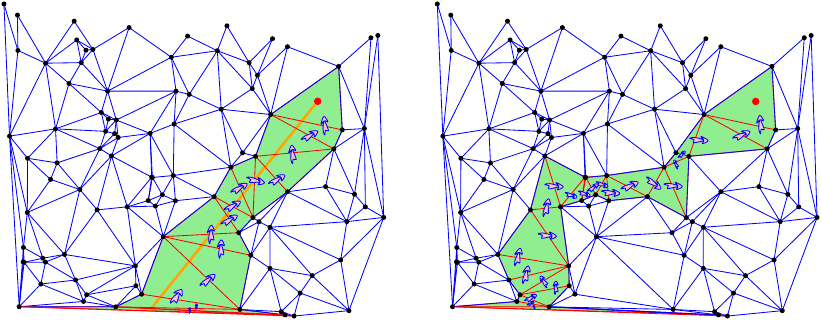}
     \caption{
Straight (left) and visibility (right) walks
        \label{fig:classics}
     }
 \end{figure}

\begin{figure}[b]
\begin{center}
\includegraphics[page=3,width=0.3\textwidth]{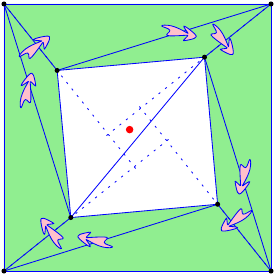}
\end{center}
\caption{
Visibility walk may loop
        \label{fig:loop}
     }
 \end{figure}

Regarding performance, both walks may visit all triangles in the worst
case and visit $O(\sqrt{n})$ triangles when the points are evenly
distributed~\cite{devroye2004expected,devillers:hal-01348831}.  From a
practical point of view, the visibility walk is simpler to implement
and a bit faster in practice because it uses less predicates to walk
from a triangle to its neighbor~\cite{devillers:hal-00850561}.

\subsection*{Contribution}

We propose \emph{celestial distance}\footnote{The name \emph{celestial
    distance} refers to the practice of \emph{celestial navigation}
  where angular distances between the celestial bodies and the horizon
  are used for navigation at sea.}  as a new way to measure the
proximity between an edge of the subdivision and a query point.  This
distance measure allows to design new walking strategies and prove
their termination.  By design these strategies are oblivious: only the
current edge of the current face determines the edges by which the
walk progresses to the next face.

Our main contribution is the \emph{celestial walk} which is a
  refinement of the well-known visibility walk that has the additional
  advantage of terminating not just on Delaunay triangulations but on
  arbitrary convex subdivisions. This is particularly useful for
  constrained and/or incremental meshing where the conditions
  necessary for termination of the visibility walk can be locally
  and/or temporarily violated.

  Another important feature of the celestial walk is that, like the
  visibility walk, it uses only orientation predicates to navigate the
  mesh. In practice, checking an orientation predicate reduces to
  computing the sign of a second degree polynomial, the degree of such
  polynomial being a relevant measure of the predicate
  complexity~\cite{liotta1998robust,boissonnat2000robust}. As a
  consequence, it becomes relatively straightforward to implement the
  celestial walk in an efficient and robust manner.

\section{Pre-requisites \label{s:prerequisites}}

Let $\mc{G} = (V, E, F)$ be a \emph{planar straight line graph}
(\PSLG) consisting of a set of vertices $V$, a set of half--edges $E$,
and a set of faces $F$.

We abstract away the borders of $\mc{G}$ by assuming that it tiles the
entire real plane. At the same time we rule out dense
  tessellations by assuming $\mc{G}$ is \emph{locally finite}
meaning the number of vertices (edges, faces) intersecting a given,
bounded area is always finite.

We assume $\mc{G}$ is given in half-edge representation. In particular
we assume the following atomic functions~\cite{bkos-cgaa-97}:
\begin{description}
\item[$\ms{origin}: E \to V$] which maps every half-edge to its start
  vertex,
\item[$\ms{target}: E \to V$] which maps every half-edge to its end
  vertex,
\item[$\ms{edge} : F \to E$] which maps every face to some edge on its
  perimeter,
\item[$\ms{face} : E \to F$] which maps every half-edge to its
corresponding (left-hand-side) face,
\item[$\ms{next} : E \to E$] which maps every half-edge to its
  successor half-edge in the counter-clockwise winding order of the
  face perimeter,
\item[$\ms{twin} : E \to E$] which maps every half-edge to its
twin half-edge running in the opposite direction, i.e.:
$\Twin{\Twin{e}}= e$,
$\Origin{e}= \Target{\Twin{e}}$ and vice versa.
\end{description}

The \emph{point location problem} in \PSLGs can now be formulated as
follows: given some goal location $p \in \mb{R}^2$ and an initial
half-edge $e_\ms{init} \in E$, find some goal edge $e_\ms{goal} \in E$
such that $p \in \Face{e_\ms{goal}}$ using only $\Next{\cdot}$ and
$\Twin{\cdot}$ to get from one half-edge to the next, i.e.: there must
exist a finite path $e_\ms{init} = e_0 \dots e_n = e_\ms{goal}$ such
that for all $0 < i \le n$ it holds $e_i = \Next{e_{i-1}}$ or $e_i
= \Twin{e_{i-1}}$.

\section{Celestial Distance}

One problem that we encounter when we try to use Euclidean distance as
a measure of progress for a walking algorithm is the fact that
Euclidean distance is not always strictly decreasing for every step in
the walk. The latter means that it is not possible to prove
termination using Euclidean distance alone. For this reason we define
the following augmented distance measure on (half-)edges.

For a given point $p \in \mb{R}^2$ and a half-edge $e \in E$ we define
the \emph{celestial distance of $e$ to $p$} as a pair $\Dist{e, p} =
[d, \alpha]$ where $d$ is the length of the line segment from $p$ to
the closest point on $e$ and $\alpha$ is the wide angle ($\alpha
\ge \frac{\pi}{2}$) between $e$ and this line segment or $0$ in case $d = 0$. We
now define a lexicographic order on celestial distances as follows:
\begin{align*}
[d, \alpha] < [d', \alpha']\text{ iff }d < d' \lor (d = d' \land
\alpha < \alpha')
\end{align*}

We illustrate this new distance measure in
Figure~\ref{fig:voronoi}. The figure shows a convex face and the
resulting partition of the plane obtained by grouping together points
based on their closest edge. This leads to two superimposed Voronoi
partitions: one based on classical Euclidean distance and one based on
celestial distance. The partition based on celestial distance is a
proper refinement of the partition based on Euclidean distance in the
sense that points which were ambiguous under Euclidean distance are
now partitionable under celestial distance.  Note that the points in
the corner-cones radiating outward from the vertices are all ambiguous
under Euclidean distance (since their closest point on the polygon is
the corner vertex which is shared by two edges) yet under celestial
distance they can be further partitioned. In particular, for the
partition based on celestial distance the ambiguous corner-cones are
split according to the angular bisectors.

\begin{figure}[t]
     \begin{center}
         \includegraphics[page=4,width=1\textwidth]{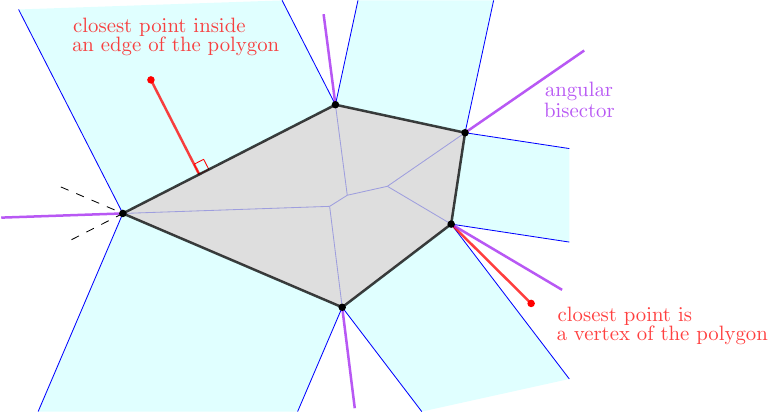}
     \end{center}
     \caption{
Voronoi diagrams of a polygon.
In blue the segment Voronoi diagram of the open edges and the vertices,
in purple the celestial Voronoi diagram of the edges.
        \label{fig:voronoi}
     }
 \end{figure}

\section{Abstract Walk}

With all the pre-requisites and our celestial distance measure in
place we are in a position to present the walking algorithm
properly. We first present an abstract version of the algorithm and
give a correctness proof for this abstract version.

In the next section, we will give an efficient, concrete instantiation
of this abstract version where the computation of the celestial
distances will be completely implicit. However, the correctness of the
final version will rest on the correctness proof of the abstract
version as given in this section.

\begin{algorithm}
\caption{Abstract Walk}\label{algo:abstract}
\begin{algorithmic}[1]
\If {$p$ strictly right of $e$}\label{yy-a}
  \State $e \gets \Twin{e}$\label{yy-b}
\EndIf\label{yy-c}
\State $E' \gets \{ e \}$\label{yy-d}
\While{$E' \ne \varnothing$}\label{yy-e}
  \State $e \gets \ms{select}(E')$\label{yy-f}
  \State $E' \gets \{ \Twin{e'}\ |\ e' \in \Face{e} \land \Dist{e',p} < \Dist{e,p} \land p\text{ strictly right of }e' \}$\label{yy-g}
\EndWhile\label{yy-h}
\State \Return{$\Face{e}$}\label{yy-i}
\end{algorithmic}
\end{algorithm}

The algorithm is rather simple: given a starting edge $e$ and a query
point $p$ to the left of $e$, we select an edge of the $\Face{e}$ that
improves the distance to $p$.
The abstract walk is formalized in Algorithm~\ref{algo:abstract}.

In lines \ref{yy-a}-\ref{yy-c} we bootstrap the invariant that the
query point is always to the left of the current edge.  In line
\ref{yy-d} we bootstrap the current set of successor candidates. In
line \ref{yy-e} we enter the main loop. The loop invariant ensures
that the loop will terminate as soon as there are no more suitable
successor edges that have lower celestial distance to the query
point. In the proof below we will see how this condition is sufficient
to ensure that, at termination, it holds $p \in \Face{e}$. In line
\ref{yy-f} we non-deterministically select one of the candidate
successor edges. In line \ref{yy-g} we compute the next set of
candidate successor edges which are all edges on the current face
perimeter that have the query point on the right and have smaller
celestial distance to the query point than the current edge.

\begin{theorem}\label{th:abstract}
For any planar subdivision, starting edge, and query point
Algorithm~\ref{algo:abstract} terminates with $p \in \Face{e}$.
\end{theorem}
\begin{proof}
For termination note that, due to local finiteness of the mesh, and
monotonicity of our distance measure, an infinitely descending chain
of celestial distances is ruled out.
It remains to prove that after termination it holds $p \in
\Face{e}$. For this it would suffice to show that for a half-edge $e$
and goal location $p$ such that $p$ is to the left-of $e$ but not in
$\Face{e}$ there always exists another $e'$ in the perimeter of
$\Face{e}$ such that $p$ is strictly on the right of $e'$ and
$\Dist{e',p} < \Dist{e,p}$.  So let $e_p$ be the point on $e$ closest
to $p$ and let $r$ be the ray from $p$ to $e_p$. We make a case
distinction on $r$. First assume that $r$ intersects the interior of
$\Face{e}$. In this case it must hold that $r$ intersects the boundary
of $\Face{e}$ at least one more time in another edge that is closer to
$p$ than the current edge for the Euclidean distance, and thus for our
celestial distance (otherwise $p$ would lie in the current face).
Next assume that $r$ does not intersect the interior of $\Face{e}$, In
this case it must hold that $e_p = \Origin{e}$ or $e_p =
\Target{e}$. Assume, w.l.o.g., that $e_p = \Target{e}$. In this case
we claim that it holds that $e' = \Next{e}$ has strictly smaller
celestial distance to $p$.  Assume, w.l.o.g., that $e'$ provides no
improvement over $e$ w.r.t. Euclidean distance to $p$.  It then
immediately follows that the closest point on $e'$ to $p$ must be the
only point that $e$ and $e'$ share which, by assumption, is $e_p =
\Target{e} = \Origin{e'}$.
Since $r$ does not intersect the boundary of $\Face{e}$ in any point
other than $e_p$ it must hold that $e'$ is in-between $e$ and $r$
(cf. Figure~\ref{fig:proof}).
Now it holds that: $\wideangle{e, r} = \narrowangle{e, e'} +
\wideangle{e', r}$ which implies $\wideangle{e',r} < \wideangle{e,
  r}$. This establishes $\Dist{e',p} < \Dist{e,p}$ and we conclude the
algorithm would progress to $\Twin{e'}$.
\end{proof}

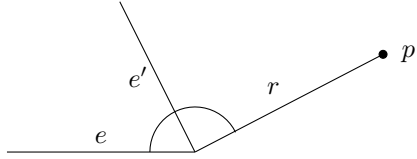
\begin{figure}[t]\centering
  \begin{tikzpicture}
    \node[] at (0, 0)   (eOrigin) {};
    \node[] at (-2.5, 0)  (eTarget) {};
    \node[] at (-1, 2)  (epTarget) {};
    \node[label=right:$p$] at (2.5, 1.3)  (p) {};

    \draw[] (eOrigin.center) -- node[above] {$e$} ++ (eTarget.center);
    \draw[] (eOrigin.center) -- node[left] {$e'$} ++ (epTarget.center);
    \draw[] (eOrigin.center) -- node[above left] {$r$} ++ (p.center);
    \filldraw[black] (2.5, 1.3) circle[radius=1.5pt];
    \draw[domain=26:180] plot ({0.6*cos(\x)}, {0.6*sin(\x)});
  \end{tikzpicture}
  \caption{Illustration for limit case of Theorem~\ref{th:abstract}.}
  \label{fig:proof}
\end{figure}

\section{Celestial Walk}

In the previous section, the successor of an edge in the walk can be the
twin of any edge of the current face with a smaller distance to
the query point.  In this section, we explain a way to actually
select such an edge using only orientation checks.

The main problems that we are solving in this section are the
  facts that Algorithm~\ref{algo:abstract} is non-deterministic and
relies on the explicit computation of celestial distances. Both these
properties make it less immediately applicable. In this section we
therefore develop a derived algorithm that works for \emph{convex}
\PSLGs. As we shall see, the additional assumption of convexity allows
us a significantly more efficient walk.

So let us first consider the problem of determining a successor edge
that has lower celestial distance than our current edge without having
to explicitly compute these distances.

As an example of a convex face consider once more
Figure~\ref{fig:voronoi}.  If we assume the query point is outside the
face this leaves two possibilities. First, the query point may be
located inside one of the orthogonal slabs (indicated in blue in
Figure~\ref{fig:voronoi}) which means the closest point coincides with
the orthogonal projection of the query point on the edge. Second, the
query point may be located in one of the intermediate corner-cones
separating the orthogonal slabs which means the closest point coincides
with the corner vertex.

The latter argument gives us a basic refinement of
Algorithm~\ref{algo:abstract} for the convex case: check if the query
point is in one of the orthogonal slabs, if so, pick the corresponding
edge, else the query point must be inside some corner-cone in-between
two orthogonal slabs, in this case we can pick either edge (unless
one of them is our current edge in which case we are forced to pick
the other alternative in order to make progress).

To do even better we can resolve the remaining non-determinism by
using the angle as a tie-breaker to determine which of the two
candidate edges is best (i.e.: which edge is most favorably oriented
towards the query point). Theoretically the best tie breaker is the
angular bisector (indicated in purple in
Figure~\ref{fig:voronoi}). The only problem is that the explicit
computation and representation of the exact angular bisector is at
least as hard as the explicit computation and representation of
celestial distances.

Fortunately we can avoid the explicit computation of angular bisectors
as well. In particular we make the following case distinction. If the
corner vertex is \emph{not obtuse} (more precisely the face has an
acute or right internal angle at that corner vertex) it holds that the
extensions of both edges lie inside the corner-cone (i.e.: the
leftmost corner in Figure~\ref{fig:voronoi}). This means that we may
use either edge itself as a crude approximation to the angular
bisector (this is precisely what the visibility walk does in
\emph{all} cases). On the other hand, if the corner vertex is
\emph{obtuse} it holds that the normal to the base of the triangle
spanned by the two candidate edges forms a suitable approximation to
the angular bisector (because it lies properly inside the
corner-cone). In Figure~\ref{fig:approx} we illustrate this
construction. Note that, in the limit, as the internal angle
approaches 180 degrees, or, alternatively, the ratio between the
lengths of the two edges approaches 1, the base normal converges to
the exact angular bisector.

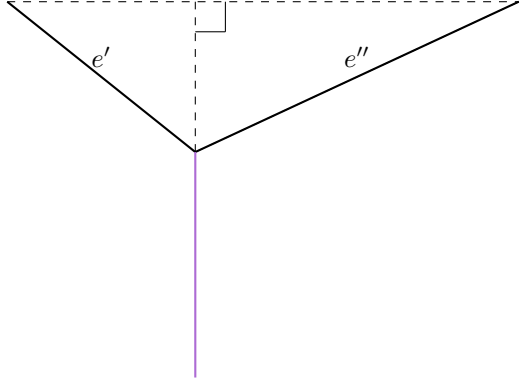
\begin{figure}\centering
  \begin{tikzpicture}
    \node[] at (0, 0)   (o) {};
    \node[] at (-2.5, 2)  (ep) {};
    \node[] at (4.3, 2)  (epp) {};
    \node[] at (0, 2) (s) {};
    \node[] at (0, -3) (e) {};

    \def \dx {0.4};
    \node[] at (\dx, 2) (a1) {};
    \node[] at (\dx, 2-\dx) (a2) {};
    \node[] at (0, 2-\dx) (a3) {};

    \draw[thick] (o.center) -- node[above] {$e'$} ++ (ep.center);
    \draw[thick] (o.center) -- node[above] {$e''$} ++ (epp.center);
    \draw[dashed] (ep.center) -- (epp.center);
    \draw[dashed] (o.center) -- (epp.center);
    \draw[dashed] (o.center) -- (s.center);
    \draw[thick,BisectorPurple] (o.center) -- (e.center);

    \draw[] (a1.center) -- (a2.center);
    \draw[] (a2.center) -- (a3.center);
  \end{tikzpicture}
  \caption{Construction for approximating angular bisector.}\label{fig:approx}
\end{figure}

\begin{figure}\centering
  \begin{tikzpicture}
    \def\xEp {-4}
    \def\yEp {2}
    \def\xEpp {3}
    \def\yEpp {2}

    \node[] at (0, 0)   (o) {};
    \node[] at (\xEp, \yEp)  (ep) {};
    \node[] at (\xEpp, \yEpp)  (epp) {};
    \def\s{1}
    \node[] at (\yEp*\s, -\xEp*\s) (s) {};

    \def\dsA{0.1}
    \def\dsN{0.101}
    \node[] at (\dsA*\xEp, \dsA*\yEp) (a1) {};
    \node[] at (\dsN*\xEp+\dsN*\yEp, \dsN*\yEp-\dsN*\xEp) (a2) {};
    \node[] at (\dsA*\yEp, -\dsA*\xEp) (a3) {};

    \draw[] (o.center) -- node[above] {$e'$} ++ (ep.center);
    \draw[] (o.center) -- node[above] {$e''$} ++ (epp.center);
    \draw[dashed] (o.center) -- (s.center);

    \draw[] (a1.center) -- (a2.center);
    \draw[] (a2.center) -- (a3.center);
  \end{tikzpicture}
  \caption{Reducing obtuseness check to orientation check.}\label{fig:obtuse}
\end{figure}
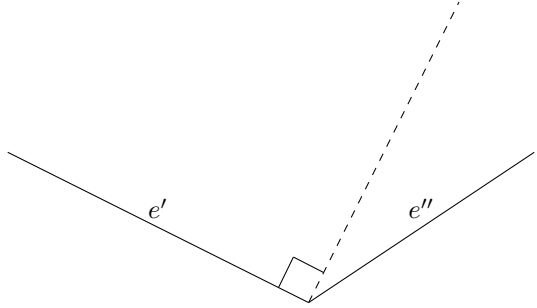

The latter considerations lead us to define $\ApproxBisector{e', e''}$
for some pair of consecutive edges $e'$, $e''$ such that $e'' =
\Next{e'}$ and $\narrowangle{e', e''}$ is obtuse. In particular, we
let $\ApproxBisector{e', e''}$ denote the line from $\Target{e'}$ ---
which is equal by definition to $\Origin{e''}$ --- in the direction
orthogonal to, and to the right of, the internal baseline segment that
connects $\Origin{e'}$ and $\Target{e''}$. For an illustration of this
construction cf. Figure~\ref{fig:approx}. We also define
  $\Obtuse{e', e''}$ as a predicate that is true iff $\Target{e''}$ is
  strictly right of the line from $\Target{e'}$ --- which we assumed
  is equal to $\Origin{e''}$ --- in the direction orthogonal to, and
  to the left of, $e'$. For an illustration of this construction
  cf. Figure~\ref{fig:obtuse}.
With these additional definitions Algorithm~\ref{algo:celestial}
formalizes the celestial walk.

\begin{algorithm}[ht]
\caption{Celestial Walk}\label{algo:celestial}
\begin{algorithmic}[1]
  \If {$p$ strictly right of $e$} \label{xx-a}
    \State $e \gets \Twin{e}$ \label{xx-b}
  \EndIf \label{xx-c}
  \State $e' \gets \Next{e}$ \label{xx-d}
  \While {$e \ne e'$} \label{xx-e}
    \If {$p$ strictly right of $e'$} \label{xx-f}
      \State $e'' \gets \Next{e'}$ \label{xx-g}
      \While {$\Obtuse{e', e''}$ and $p$ left of $\ApproxBisector{e', e''}$} \label{xx-h}
        \State $e', e'' \gets e'', \Next{e''}$ \label{xx-i}
      \EndWhile \label{xx-j}
      \State $e \gets \Twin{e'}$ \label{xx-k}
      \State $e' \gets \Next{e}$ \label{xx-l}
    \Else \label{xx-m}
      \State $e' \gets \Next{e'}$ \label{xx-n}
    \EndIf \label{xx-o}
  \EndWhile \label{xx-p}
  \State \Return{$\Face{e}$} \label{xx-q}
\end{algorithmic}
\end{algorithm}

In lines \ref{xx-a}-\ref{xx-c} we bootstrap the invariant that the
query point is always to the left of the current edge. In line
\ref{xx-d} we initialize a perimeter edge variable used for iterating
over the outline of the current face. In line \ref{xx-e} we enter the
outer loop. As can be seen by inspecting the code-paths inside the
loop body, the loop invariant ensures that the loop will terminate iff
the query point is inside the current face. In line \ref{xx-f} we
check if the query point lies strictly to the right of the perimeter
edge. If this is the case we know, due to the face convexity,
that the query point is outside the face. From that moment on
the only goal is to find an edge that allows us to walk to a
neighboring face. Therefore, in line \ref{xx-g} we introduce a second
perimeter edge variable that will be used to iterate over \emph{pairs}
$(e', e'')$ of consecutive edges rather than singletons. In line
\ref{xx-h} we state the negation of the condition that we are looking
for: either the internal angle of the face is acute\footnote{Note that
  the obtuseness check in line \ref{xx-h} can potentially be memorized
  by the mesh at the expense of one bit per half-edge.} or, otherwise,
the query point falls strictly to the right of the approximate
bisector. In line \ref{xx-i} we shift the pair of edges along the face
perimeter until said condition is reached. In line \ref{xx-k} we have
exited the inner loop so we know that $e'$ contains the edge that is
``sufficiently closest'' to the query point for us to continue the
walk to the neighboring face. Hence we re-assign $e$ and $e'$ and
drop back into the outer loop as though we would have restarted the
algorithm on this new edge.

The latter can be seen as a form of tail-recursion. Indeed, a
recursive formulation is trivially possible by transforming the entire
procedure into a function and substituting a recursive call at
line~\ref{xx-l}. We have not opted for the recursive formulation
because it would be slightly less efficient. The efficiency loss would
be due to the bootstrapping check in lines \ref{xx-a}-\ref{xx-c}
which, in the recursion, would be superfluous.

\section{Conclusion}

We have shown how our celestial distance measure allows to design
novel walking strategies. In particular we improved the applicability
of the popular visibility walk by refining it into what we call the
celestial walk. Our walk is almost as simple to implement yet is
guaranteed to terminate on arbitrary convex subdivisions.

If we assume that all obtuseness checks are pre-computed or memoized
(at the expense of 1 or 2 bits per half-edge respectively) we can make
the following observations w.r.t. the expected complexity of our walk.

The worst case for our walk materializes when all angles in the mesh
are obtuse, i.e. in a regular hexagonal mesh, in this case we require
2 orientation tests per visited half-edge (the same number as the
straight walk).

For triangulations a simple counting argument proves that there is at
most a fraction of $\frac{1}{3}$ of obtuse angles. This gives us a
worst case rough estimate of $\frac{4}{3}$ orientation tests per
visited half-edge (better than the straight walk and only slightly
worse than the visibility walk). In practice, we expect better
performance since the fraction of obtuse angles will be much less than
$\frac{1}{3}$ in triangulations that exhibit some regularity.

For meshes that are guaranteed to be triangular, minor optimizations
are possible by unrolling the loops of the general
algorithm. Furthermore, it is a good heuristic to minimize the number
of visited triangles by alternating clockwise and counterclockwise
winding order in order to balance out the systematic errors in
approximating the angular bisectors.

\subsection*{Future Work}

An obvious next question to consider is whether the walk can be
generalized to 3 dimensions. We have some reason to believe this is
the case.

In 2D as we progress according to the celestial walk we either move
forward in the direction of the query point or \emph{orient} ourselves
towards it.  
Similarly, in 3 dimensions we expect to either progress in Euclidian
sense or \emph{roll} in the direction of the query point.  
This intuition leads to developing a suitable non-increasing distance
measure quantifying this relationship.
In particular to the walking problem in tetrahedralization we find
that a combination of Euclidian distance as the first progression
criterion and two successive angles are sufficient to guarantee the
progress towards the tetrahedron containing the query
point~\cite{kuijper2017zig}.

Finally, we would also be interested in obtaining experimental results
regarding the trade-off that exists between steering (i.e.: choosing a
next edge that is favorably oriented to the query point) and driving
(i.e.: choosing a next edge \emph{quickly} based on easily computed
criteria). 

\small

\bibliographystyle{abbrv}
\bibliography{biblio}

\begin{thebibliography}{10}

\bibitem{boissonnat2000robust}
J.-D. Boissonnat and F.~P. Preparata.
\newblock Robust plane sweep for intersecting segments.
\newblock {\em SIAM Journal on Computing}, 29(5):1401--1421, 2000.

\bibitem{bose2002online}
P.~Bose, A.~Brodnik, S.~Carlsson, E.~D. Demaine, R.~Fleischer,
  A.~L{\'o}pez-Ortiz, P.~Morin, and J.~Ian~Munro.
\newblock Online routing in convex subdivisions.
\newblock {\em International Journal of Computational Geometry \&
  Applications}, 12(04):283--295, 2002.

\bibitem{bose2004online}
P.~Bose and P.~Morin.
\newblock Online routing in triangulations.
\newblock {\em SIAM journal on computing}, 33(4):937--951, 2004.

\bibitem{chazelle1986fractional}
B.~Chazelle and L.~J. Guibas.
\newblock Fractional cascading: I. a data structuring technique.
\newblock {\em Algorithmica}, 1(1):133--162, 1986.

\bibitem{bkos-cgaa-97}
M.~de~Berg, M.~van Kreveld, M.~Overmars, and O.~Schwarzkopf.
\newblock {\em Computational Geometry: Algorithms and Applications}.
\newblock Springer-Verlag, Berlin, 1997.

\bibitem{devillers:inria-00166711}
O.~Devillers.
\newblock {The Delaunay Hierarchy}.
\newblock {\em {International Journal of Foundations of Computer Science}},
  13:163--180, 2002.

\bibitem{devillers:hal-00850561}
O.~Devillers.
\newblock {Delaunay triangulations, theory vs practice.}
\newblock In {\em {EuroCG, 28th European Workshop on Computational Geometry}},
  Assisi, Italy, 2012.

\bibitem{devillers:hal-01348831}
O.~Devillers and R.~Hemsley.
\newblock {The worst visibility walk in a random Delaunay triangulation is
  $O(\sqrt{n})$ }.
\newblock {\em {Journal of Computational Geometry}}, 7(1):332--359, 2016.

\bibitem{devillers:inria-00102194}
O.~Devillers, S.~Pion, and M.~Teillaud.
\newblock {Walking in a Triangulation}.
\newblock {\em {International Journal of Foundations of Computer Science}},
  13:181--199, 2002.

\bibitem{devroye2004expected}
L.~Devroye, C.~Lemaire, and J.-M. Moreau.
\newblock Expected time analysis for delaunay point location.
\newblock {\em Computational geometry}, 29(2):61--89, 2004.

\bibitem{kirkpatrick1983optimal}
D.~Kirkpatrick.
\newblock Optimal search in planar subdivisions.
\newblock {\em SIAM Journal on Computing}, 12(1):28--35, 1983.

\bibitem{kuijper2017zig}
W.~Kuijper, V.~Ermolaev, and H.~Berntsen.
\newblock Zig-zagging in a triangulation.
\newblock {\em arXiv preprint arXiv:1705.03950}, 2017.

\bibitem{liotta1998robust}
G.~Liotta, F.~P. Preparata, and R.~Tamassia.
\newblock Robust proximity queries: An illustration of degree-driven algorithm
  design.
\newblock {\em SIAM Journal on Computing}, 28(3):864--889, 1998.

\bibitem{mucke1999fast}
E.~P. M{\"u}cke, I.~Saias, and B.~Zhu.
\newblock Fast randomized point location without preprocessing in two-and
  three-dimensional {D}elaunay triangulations.
\newblock {\em Computational Geometry}, 12(1-2):63--83, 1999.

\bibitem{preparata1990planar}
F.~P. Preparata.
\newblock Planar point location revisited.
\newblock {\em International Journal of Foundations of Computer Science},
  1(01):71--86, 1990.

\end{thebibliography}

\end{document}